\documentclass[11pt]{article}
\usepackage[letterpaper, left=1in, right=1in, bottom=1in, top=1in]{geometry}
\usepackage{authblk}

\usepackage[dvipdfmx]{graphicx}
\usepackage{array}
\usepackage{ascmac}
\usepackage{amsmath,amssymb,amsthm}
\usepackage{booktabs}
\usepackage[dvipdfmx]{graphicx}
\usepackage{wrapfig}
\usepackage[algoruled,linesnumbered,algo2e,vlined]{algorithm2e}
\usepackage{url}
\usepackage{multirow}
\usepackage{multicol}
\usepackage{cases}
\usepackage[normalem]{ulem}
\usepackage{boites}
\usepackage{hyperref}

\newcommand{\Sort}{\mathsf{sort}}
\newcommand{\Prefix}{\mathsf{Prefix}}
\newcommand{\Substr}{\mathsf{Substr}}
\newcommand{\Suffix}{\mathsf{Suffix}}
\newcommand{\rev}[1]{#1^R}
\newcommand{\lrp}{\mathsf{lrp}}
\newcommand{\lrs}{\mathsf{lrs}}
\newcommand{\STree}{\mathsf{STree}}
\newcommand{\BPSTree}{\mathsf{BPSTree}}

\newcommand{\BP}{\mathsf{BP}}
\newcommand{\str}{\mathsf{str}}
\newcommand{\RLPF}{\mathsf{RLPF}}
\newcommand{\LPF}{\mathsf{LPF}}
\newcommand{\LCF}{\mathsf{LCF}}
\newcommand{\LCFA}{\mathsf{LCFA}}
\newcommand{\mcf}{\mathsf{mcf}}
\newcommand{\MCFA}{\mathsf{MCFA}}
\newcommand{\ins}{\mathsf{insert}}
\newcommand{\delete}{\mathsf{delete}}
\newcommand{\update}{\mathsf{update}}
\newcommand{\RmQ}{\mathsf{RmQ}}
\newcommand{\tree}{\mathsf{T}}
\newcommand{\weight}{\mathsf{weight}}

\newtheorem{theorem}{Theorem}
\newtheorem{corollary}{Corollary}
\newtheorem{lemma}{Lemma}

\newtheorem{definition}{Definition}

\newtheorem{example}{Example}

\usepackage[final,mode=multiuser]{fixme}

\usepackage{colortbl}
\fxsetup{theme=color}
\definecolor{fxtarget}{rgb}{0.0000,0.0000,0.4823}
\fxuseenvlayout{colorsig}
\fxusetargetlayout{color}
\fxuseenvlayout{colorsig}
\fxusetargetlayout{color}
\FXRegisterAuthor{ws}{aws}{WS}
\FXRegisterAuthor{tm}{atm}{TM}
\FXRegisterAuthor{si}{asi}{SI}
\fxsetface{env}{}

\begin{document}
\title{Faster and simpler online/sliding rightmost Lempel-Ziv factorizations}
\author[1]{Wataru~Sumiyoshi}
\author[2]{Takuya~Mieno}
\author[1]{Shunsuke~Inenaga}

\affil[1]{Kyushu University, Japan}
\affil[ ]{\texttt{sumiyoshi.wataru.342@s.kyushu-u.ac.jp}}
\affil[ ]{\texttt{inenaga.shunsuke.380@m.kyushu-u.ac.jp}}
\affil[2]{University of Electro-Communications, Japan}
\affil[ ]{\texttt{tmieno@uec.ac.jp}}

\date{}
\maketitle

\begin{abstract}
  We tackle the problems of computing the \emph{rightmost} variant of
  the Lempel-Ziv factorizations in the online/sliding model.
  Previous best bounds for this problem are $O(n \log n)$ time
  with $O(n)$ space, due to Amir et al.~[IPL 2002] for the online model,
  and due to Larsson~[CPM 2014] for the sliding model.
  In this paper, we present faster $O(n \log n / \log \log n)$-time
  solutions to both of the online/sliding models.
  Our algorithms are built on a simple data structure
  named \emph{BP-linked trees},
  and on a slightly improved version of the range minimum/maximum query (RmQ/RMQ) data structure on a dynamic list of integers.
  We also present other applications of our algorithms.
\end{abstract}

\section{Introduction}

\subsection{Online rightmost LZ-factorizations and LPF arrays}

The \emph{longest previous factor array}\footnote{Our definition of online LPF arrays follows from the literature~\cite{OkanoharaS08,PrezzaR20}.} $\LPF$ of a string $S$ of length $n$
is an array of length $n$ such that, for each $1 \leq i \leq n$, $\LPF[i]$ stores the length $\ell_i$ of the longest suffix of $S[1..i]$ that occurs at least twice in $S[1..i]$.
The LPF array has a close relationship to the Lempel-Ziv (LZ) factorization~\cite{LempelZ76}, that is a basic and powerful tool for a variety of string processing tasks including data compression~\cite{LZ77} and finding repetitions~\cite{KolpakovK99}.

We consider a variant of LPF arrays with \emph{rightmost reference},
denoted $\RLPF$,
where each $\RLPF[i]$ also stores
the distance $d = i-j$ to the rightmost previous ending position $j$~($j < i$)
of the longest repeating length-$\ell_i$ suffix of $S[1..i]$.
Computing the rightmost references is motivated by encoding
each factor in the LZ-factorization with less bits~\cite{FerraginaNV13},
and has attracted much attention.
The state-of-the-art \emph{offline} algorithm for the rightmost LZ-factorization
runs in
$O(n(\log\log\sigma + \frac{\log \sigma}{\sqrt{\log n}}))$ time
with $O(n \log \sigma)$ bits of space, where $\sigma$ is the alphabet size~\cite{BelazzouguiP16}.
Bille et al.~\cite{BilleCFG17} proposed an algorithm for computing
a $(1+\epsilon)$-approximated version of the rightmost LZ-factorization for any $\epsilon > 0$.
Ellert et al.~\cite{EllertFP23} considered
the rightmost version of the \emph{LZ-End} factorization~\cite{KreftN10},
a variant of the LZ-factorization designed for fast random access.

The other common method for limiting the distance from each factor to a previous occurrence is the \emph{sliding} model,
where only the previous occurrences of each factor within the preceding sliding window of fixed size $d \geq 1$ are considered~\cite{StorerS82,Bell86}.
The LZ-factorization in the sliding model is used in the real-world compression software's including \texttt{zip} and \texttt{7zip}. 
Sliding suffix tree algorithms~\cite{Larsson96,Senft2005,leonard2024constanttime} are able to compute the LZ-factorization in the sliding model in
$O(n \log \sigma)$ time with $O(d)$ words of working space.
Bille et al.~\cite{BilleCFG17} presented
another algorithm for sliding LZ-factorization that runs in $O(\frac{n}{d} \Sort(d) + z \log \log \sigma)$ time with $O(d)$ words of working space,
where $z$ is the number of factors and 
$\Sort(d)$ denotes the time for sorting the $d$ characters
in each of the $O(\frac{n}{d})$ blocks on the input string.

In this paper, we consider the three following problems:
\begin{description}
  \item[Problem (1):] The rightmost LPF array in the online model.
  \item[Problem (2):] The rightmost LZ-factorization in the online model.
  \item[Problem (3):] The rightmost LZ-factorization in the sliding model.
\end{description}

Amir et al.~\cite{timestamped_suffixtree} proposed
an algorithm for (1) that works in $O(n \log n)$ time with $O(n)$ words of space.
Their key data structure is the \emph{timestamped suffix tree}, 
which is based on Weiner's online suffix tree construction~\cite{Weiner73}
and is augmented with an online range minimum query data structure.
Larsson~\cite{Larsson14} presented an algorithm
for (2) running in $O(n \log n)$ time with $O(n)$ words of space,
that is based on Ukkonen's online suffix tree construction~\cite{Ukkonen95}.
To the best of our knowledge,
none of the existing algorithms provides an efficient solution to (3),
where \emph{both} of the rightmost and sliding properties are required.

\subsection{Our new online/sliding algorithms for rightmost LZ and LPF}

We consider a simple data structure named \emph{BP-linked trees}
capable of maintaining a representation of balanced parentheses (BP) of a dynamic rooted tree.
{Basically, our BP-linked trees are equivalent to an intermediate
  data structure used in 
  the so-called \emph{Euler tour trees}~\cite{HenzingerK99} that maintain the Euler tours of dynamic trees:
  Our BP-linked trees can be seen as a representation of the Euler tours
  of the input trees.
}In our BP-linked tree, the BP is maintained as a doubly-linked list,
which can be updated in $O(1)$ worst-case time
given the locus of the inserted/deleted node on the explicitly stored tree.
By maintaining our BP-linked tree on top of the suffix tree,
we achieve an online algorithm for computing rightmost LPF arrays in 
$O(n \log n / \log \log n)$ time with $O(n)$ words of space,
thus achieving a faster online solution for (1).
In addition, we show how our algorithm can be modified to solve (2)
in the same complexity as (1),
and in $O(n \log d / \log \log d)$ time with $O(d)$ words of working space
for (3).

The $\log n / \log \log n$ (resp. $\log d / \log \log d$) term
in our time complexities comes from \emph{range minimum/maximum queries} (\emph{RmQ/RMQ})
on a dynamic list of $n$ integers (resp. $d$ integers) -
to compute the rightmost LZ-factorization and LPF array,
we use RmQ/RMQ to retrieve the rightmost previous occurrence
of a given locus in the online/sliding suffix tree.
While those bounds for dynamic RmQ/RMQ can already be achieved by the use of Brodal et al.'s \emph{path minimum/maximum queries} data structure on a dynamic tree~\cite{BrodalDR11} \emph{in the amortized sense}, this paper shows how their data structure can be modified to perform updates and queries in the same \emph{worst-case time bounds} in the case of dynamic lists, after sublinear-time preprocessing (Lemma~\ref{lem:dynamic_rmq}).

The simple framework of our algorithms allows one to obtain
very simple alternative solutions to the existing ones:
By using folklore dynamic RmQ/RMQ data structures based on binary search trees
in place of the aforementioned advanced RmQ/RMQ data structures,
the same run times as the methods of Amir et al.~\cite{timestamped_suffixtree}
for (1) and Larsson~\cite{Larsson14} for (2) can readily be achieved.
{It appears that this version of our BP-linked trees with binary search trees
  is basically equivalent to the so-called Euler tour trees~\cite{HenzingerK99}
  that support updates and queries on dynamic input trees in $O(\log n)$ time each.
}

We also present other applications of our algorithms in Section~\ref{sec:applications}.

\subsection{Related work for dynamic BP maintenance}

In the problem of maintaining the BP $\mathcal{B}$ for a \emph{dynamic} tree,
one is required to efficiently support the following operations and queries:
\begin{itemize}
  \item insert: add a new node to $\mathcal{B}$;
  \item delete: remove an existing non-root node from $\mathcal{B}$;
  \item leftmost leaf: return the left parenthesis ``$($'' corresponding to a given node;
  \item rightmost leaf: return the right parenthesis ``$)$'' corresponding to a given node;
  \item parent: return the nearest enclosing parentheses for a given node;
  \item rank $i$: return the number of left/right parentheses in $\mathcal{B}[1..i]$;
  \item select $i$: return the $i$th left/right parenthesis in $\mathcal{B}$.
\end{itemize}
This problem was already studied at least in early 80's,
in the context of maintaining a dynamic set of nesting intervals~\cite{GutingW82}.
Since then, it has also appeared in various important problems including dynamic dictionary matching~\cite{AmirFIPS95,ChanHLS07} and (compressed) suffix trees of dynamic collection of strings~\cite{AmirFIPS95,ChanHLS07,NavarroS14}.

Navarro and Sadakane~\cite{NavarroS14} proposed a data structure of $2n + o(n)$ bits of space that supports all the above queries and operations in worst-case $O(\log n / \log \log n)$ time.
Chan et al.~\cite{ChanHLS07} showed an amortized $\Omega(\log n / \log \log n)$-time lower bound for the dynamic BP-maintenance via a reduction from the dynamic subset rank problem on a set $\mathcal{S}$ of integers~\cite{FredmanS89}.
Chan et al. reduce a subset rank query on $\mathcal{S}$ to finding the nearest enclosing parentheses in $\mathcal{B}$ (i.e. finding the parent node),
which can further be reduced to a constant number of rank/select queries in $\mathcal{B}$.
Thus, any algorithm for dynamic BP-maintenance \emph{which supports rank/select queries} must use (amortized) $\Omega(\log n / \log \log n)$ time.

Our BP-linked trees
deal with a simpler version of the dynamic BP-maintenance problem where all the operations and queries, \emph{excluding rank and select queries}, are supported.
Our BP-linked trees are a simple pointer-based data structure,
which occupies $O(n)$ words of space and performs insertions, deletions,
accessing the leftmost/rightmost leaf, and the parent, in worst-case $O(1)$ time each.

\section{Preliminaries} \label{sec:preliminaries}

\subsection{Strings}
Let $\Sigma$ denote an ordered \emph{alphabet} of size $\sigma$.
An element of $\Sigma^*$ is called a \emph{string}.
The length of a string $S \in \Sigma^*$ is denoted by $|S|$.
The \emph{empty string} $\varepsilon$ is the string of length $0$.
For string $S = xyz$, $x$, $y$, and $z$ are called
the \emph{prefix}, \emph{substring}, and \emph{suffix} of $S$,
respectively.
Let $\Prefix(S)$, $\Substr(S)$, and $\Suffix(S)$ denote
the sets of prefixes, substrings, and suffixes of $S$, respectively.
For a string $S$ of length $n$, $S[i]$ denotes the $i$th symbol of $S$
and $S[i..j] = S[i] \cdots S[j]$ denotes the substring of $S$
that begins at position $i$ and ends at position $j$ for $1 \leq i \leq j \leq n$.
For convenience, let $S[i..j] = \varepsilon$ for $i > j$.
The \emph{reversed string} of a string $S$ is denoted by $\rev{S}$,
that is, $\rev{S} = S[|T|] \cdots S[1]$.

For a string $S$,
the strings in $\Prefix(S) \cap \Substr(S[2..|S|])$
and the strings in $\Suffix(S) \cap \Substr(S[1..|S|-1])$
are called \emph{repeating prefixes} and \emph{repeating suffixes}
of $S$, respectively.
Let $\lrp(S)$ and $\lrs(S)$ denote the longest repeating prefix
and the longest repeating suffix of $S$, respectively.

\subsection{Model of computation}
This paper assumes the standard \emph{word RAM model} with word size $\Theta(\log n)$, where $n$ is the length of the input string.

\subsection{Suffix trees}
The \emph{suffix tree}~\cite{Weiner73} of a string $S$,
denoted $\STree(S)$, is a path-compressed trie representing $\Suffix(S)$ such that
\begin{itemize}
  \item[(1)] Each internal node has at least two children;
  \item[(2)] Each edge is labeled by a non-empty substring of $S$;
  \item[(3)] The labels of out-going edges of the same node begin with distinct characters.
\end{itemize}
Each leaf of $\STree(S)$ is associated with the beginning position
of its corresponding suffix of $S$.
For a node $v$ of $\STree(S)$, let $\str(v)$ denote the string label
of the path from the root to $v$.
Each node $v$ stores its string depth $|\str(v)|$.
The \emph{locus} of a substring $w \in \Substr(S)$ in $\STree(S)$
is the position where $w$ is spelled out from the root.
The locus of $w$ is said to be an \emph{explicit node}
if $w = \str(v)$ for some node $v$ in $\STree(S)$.
Otherwise, i.e. the locus of $w$ is on an edge,
then it is said to be an \emph{implicit node}.
The number of explicit nodes in $\STree(S)$ is at most $n-1$,
where $n = |S|$,
while there are $O(n^2)$ implicit nodes in $\STree(S)$.
We can represent $\STree(S)$ in $O(n)$ space
by representing each edge label $x$ with a pair $(i,j)$
of positions in $S$ such that $S[i..j] = x$.

\subsection{Online/sliding rightmost LPF arrays and LZ-factorizations}
The \emph{online longest previous factors problem} is,
given the $i$th character $S[i]$ of an online input string $S$,
to compute the longest suffix $S[i-\ell_i+1..i]$ of $S[1..i]$
that occurs at least twice in $S[1..i]$.
The \emph{rightmost longest previous factor array} of a string $S$ of length $n$,
denoted $\RLPF$, is an array of length $n$
such that
\[
  \RLPF[i] =
  \begin{cases}
    (0,1) & \mbox{if $i$ is the first occurrence of character $S[i]$ in $S$}\\
    (\ell_i, i-j) & \mbox{otherwise},
  \end{cases}
\]
where $\ell_i = |\lrs(S[1..i])|$ and 
$j = \max\{j' \mid S[i-\ell_i+1..i] = S[j'-\ell_i+1..j'], j' < i\}$.

A sequence $S = f_1, \ldots, f_z$ of $z$ non-empty strings
is called the \emph{Lempel-Ziv} (\emph{LZ}) factorization of string $S$ of length $n$
if (1) $f_k$ is a fresh character not occurring to its left in $S$,
or (2) 
$f_k$ is the longest prefix of the suffix $f_k \cdots f_z = S[|f_1 \cdots f_{k-1}|+1..n]$ of $S$ that has a previous occurrence beginning in $f_1 \cdots f_{k-1} = S[1..|f_1 \cdots f_{k-1}|]$.
In the \emph{rightmost} LZ-factorization of $S$,
each factor $f_k$ of type (2) is encoded by
a pair $(|f_k|, x)$ such that
$x = |f_1 \cdots f_k|-j$ is the distance to the ending position $j$ of the rightmost previous occurrence of $f_k$ in $S[1..|f_1 \cdots f_k|]$.

\begin{example}
  The following table shows $\RLPF$ of string $S =\mathtt{abaababaabba}$:
  \begin{center}
    \begin{tabular}{lcccccccccccc} \hline
      $i$        & 1            & 2            & 3            & 4            & 5            & 6            & 7            & 8            & 9            & 10           & 11           & 12           \\\hline
      $S[i]$     & $\mathtt{a}$ & $\mathtt{b}$ & $\mathtt{a}$ & $\mathtt{a}$ & $\mathtt{b}$ & $\mathtt{a}$ & $\mathtt{b}$ & $\mathtt{a}$ & $\mathtt{a}$ & $\mathtt{b}$ & $\mathtt{b}$ & $\mathtt{a}$ \\
      $\RLPF[i]$ & (0,1)        & (0,1)        & (1,2)        & (1,1)        & (2,3)        & (3,3)        & (2,2)        & (3,2)        & (4,5)        & (5,5)        & (1,1)        & (2,4)        \\\hline
    \end{tabular}
  \end{center}
  The rightmost LZ-factorization of $S$ is
  $(0,\mathtt{a}),(0,\mathtt{b}),(1,2),(3,3),(4,5),(2,4)$.
\end{example}

Let $d \geq 1$ denote the window size of fixed length.
A sequence $S = g_1, \ldots, g_m$ of $m$ non-empty strings
is called the \emph{sliding LZ-factorization} of a string
$S$ of length $n$ w.r.t. window size $d$,
if each factor $g_k$ is the longest prefix of the suffix
$|g_k \cdots g_m| = S[|g_1 \cdots g_{k-1}|+1..n]$ of $S$
that has a previous occurrence beginning in the sliding window
$W_k = S[\max\{1, |g_1 \cdots g_{k-1}|-d+1\}.. |g_1 \cdots g_{k-1}|]$.

\section{Data structures}

This section introduces data structures for dynamic trees
which are core components of our rightmost LZ algorithms.

\subsection{BP-linked trees} \label{sec:dynamicBP}

Let $\tree$ be a rooted ordered tree having $N$ nodes.
Let $\BP(\tree) \in \{(,)\}^{2N}$ be the BP-representation of $\tree$.
In this paper, we implement $\BP(\tree)$ using a doubly-linked list.
For each node $v$ in $\tree$, let $(_v$ and $)_v$ denote the $($ and $)$
that correspond to $v$ in $\BP(\tree)$.
A \emph{BP-linked tree} is a tree $\tree$ augmented with its BP-representation $\BP(\tree)$
such that
each node $v$ of $\tree$ has pointers to $(_v$ and $)_v$ in $\BP(\tree)$.

We consider the following edit operations on $\tree$:
(1) inserting a leaf,
or a new root as the parent of the old root,
(2) inserting an internal node by splitting an edge, and
(3) deleting a non-root node.
We remark that our tree $\tree$ is explicitly stored,
and the input of each operation is given as a locus on the tree $\tree$
(not on $\BP(\tree)$).
The next lemma follows:
\begin{lemma} \label{lem:bptree}
  Given a tree-editing operation,
  we can update a BP-linked tree in worst-case $O(1)$ time.
\end{lemma}
\begin{proof}
  First we consider the case where a leaf $v$ is inserted.
  Let $u$ be the parent of $v$.
  If $v$ is the leftmost child of $u$,
  then we take the pointer of $u$ to access $(_u$ in $\BP(\tree)$,
  and then insert $(_v$ and $)_v$ immediately to the right of $(_u$.
  Otherwise, let $x$ be $v$'s neighbor to the left.
  Then, in a similar way as before, insert $(_v$ and $)_v$ immediately to the right of $)_x$.
  Also, when a new root $r$ is inserted, we just prepend $(_r$ and append $)_r$ to $\BP(\tree)$.

  Second we consider the case where an internal node $v$ is inserted.
  Suppose that an edge $e = (u,w)$ is split into two edges
  $e_1 = (u, v)$ and $e_2 = (v, w)$.
  We take the pointer of $w$ to access $(_w$ in $\BP(\tree)$,
  and insert $(_v$ immediately to the left of $(_w$.
  We also take the right pointer of $w$ to access $)_w$ in $\BP(\tree)$,
  and then insert $)_v$ immediately to the right of $)_w$.

  Third we consider the case where a non-root node $v$ is deleted.
  Then we just delete $(_v$ and $)_v$ from $\BP(\tree)$.
  Note that if $u$ is the parent of $v$ and $v$ has $k$ children $w_1, \ldots, w_k$,
  then new parent of $w_1, \ldots, w_k$ becomes $u$ after the deletion.

  It is clear that each of these operations takes $O(1)$ worst-case time.
\end{proof}

\subsection{Subtree minimum queries}

In this subsection,
we propose dynamic data structures with worst-case update/query time
for \emph{range minimum queries (RmQs)} on a linear list and
for \emph{subtree minimum queries (SmQs)} on a rooted and weighted tree.

\subsubsection{Dynamic range minimum queries.}
A \emph{dynamic range minimum query} (\emph{RmQ}) data structure on a linear-linked-list of integers supports the following:
\begin{itemize}
  \item $\ins(u, v, x)$: insert a new node $v$ with value $x$ as the next node of $u$;
  \item $\delete(v)$: delete node $v$ from the list;
  \item $\update(v, x)$: update the value of node $v$ to $x$;
  \item $\RmQ(u, v)$: return a node with the smallest value in the path $(u, v)$.
\end{itemize}

Brodal et al.~\cite{BrodalDR11} presented a dynamic RmQ data structure
for a linear-linked-list\footnote{They actually presented a data structure for \emph{Path Minimum Queries} for an edge-weighted dynamic tree,
  which is a generalization of RmQs for a dynamic linear list.
  Since such a general setting is not needed for our purpose,
  we cite their result as a dynamic RmQ data structure and make some changes to it for simplicity.
} of $n$ integers,
which takes $O(n)$ space and supports the above queries and updates in \emph{amortized} $O(\log n / \log \log n)$ time each
in the RAM model. Below we make a few changes to their method in order to obtain \emph{worst-case} time guarantees:

\begin{lemma} \label{lem:dynamic_rmq}
  After $o(n)$-time preprocessing,
  we can maintain a dynamic RmQ data structure on a linear-linked-list of $n$ integers
  which takes $O(n)$ space and supports each query/operation in worst-case $O(\log n / \log \log n)$ time.
\end{lemma}
\begin{proof}
  Let $L$ be the dynamic list of integers.
  Let $B = \lfloor \log^\varepsilon n\rfloor \ge 1$ for some small constant $0 < \varepsilon < 1$.
  We build a \emph{q*-heap}~(Corollary~3.4 of~\cite{Willard00}) on top of the dynamic list $L$,
  which is a variant of B-trees of order $B$ and supports predecessor queries, insertions, and deletions over $L$
  in \emph{worst-case} $O(\log n / \log \log n)$ time each, after $o(n)$-time preprocessing.
  Note that updating a value of an element in $L$ can be simulated by combining an insertion and a deletion.
  Also, 
  as in {Theorem~2 of}~\cite{BrodalDR11}, we precompute lookup-tables of total size $o(n)$
  in order to support $\RmQ$, $\ins$, $\delete$ and $\update$ inside any list of size $O(B)$,  which represents a node of the q*-heap, in worst-case $O(1)$ time
  {in the RAM model.}
  Then we maintain, for each node of the q*-heap, the list consisting of the minima of its children by using the lookup-tables.
  Given a range minimum query, we can answer the query by visiting at most $O(\log n / \log \log n)$ nodes of the q*-heap,
  similar to the standard method for 1D-range trees~(see~\cite{MBCT2023} for example).
\end{proof}

\subsubsection{Dynamic subtree minimum queries.}

We introduce {subtree minimum queries (SmQs)}
on a rooted and weighted tree.
\begin{definition}
  A subtree minimum query (SmQ) on a rooted and weighted tree $\tree$ is, given a node $v$ in $\tree$,
  to compute a node having the minimum weight in the subtree rooted at $v$.
\end{definition}
For the static case,
we can easily answer any query in constant time
after storing the answer to each node by traversing the tree.

We focus on a dynamic case,
where tree-editing operation mentioned in Section~\ref{sec:dynamicBP} will be applied to the tree.
Furthermore, we consider update operations, i.e.,
updating the weight of a node to a new weight.
We show the next lemma.

\begin{lemma} \label{lem:smq}
  After $o(n)$-time preprocessing,
  we can maintain a dynamic SmQ data structure on a rooted and weighted tree with $n$ nodes
  which takes $O(n)$ space and supports each query/operation in worst-case $O(\log n / \log \log n)$ time.
  Also, the time complexity per each query/operation is optimal.
\end{lemma}
\begin{proof}
  Let $\tree$ be the input tree.
  Further let $\weight(v)$ be the weight of $v$ for each node $v$ in $\tree$.
  The SmQs on $\tree$ can be reduced to the RmQs on $\BP(\tree)$ as follows:
  For each node $v$ of $\tree$, the weight of "$(_v$" is assigned $\weight(v)$ and
  the weight of "$)_v$" is assigned $\infty$.
  By doing this reduction, it follows that
  for any node $v$ in $\tree$,
  if RmQ for pair "$(_v$", "$)_v$" returns "$(_u$", then
  node $u$ is an answer of SmQ for $v$.
  Since we can maintain $\tree$ as a BP-linked tree for any given tree-editing operation in $O(1)$ time~(Lemma~\ref{lem:bptree}),
  we can maintain the $\BP(\tree)$ with weights in $O(1)$ time as well.
  Also, by Lemma~\ref{lem:dynamic_rmq},
  the RmQ data structure on $\BP(\tree)$ can be maintained in
  worst-case $O(\log n / \log \log n)$ time for each query/editing operation.
  Therefore, we obtain the desired upper bound.

  To prove the lower bound, we reduce the \emph{priority searching problem}~\cite{AlstrupHR98} to the dynamic SmQ problem.
  Let $S \subseteq \{1, \ldots, n\}$ be a set of integers with priorities.
  A priority $p(x)$ of an integer $x$ is a positive integer at most $n$.
  The priority searching problem on $S$ supports
  (1) insertion of an integer $x$ with priority $p(x)$ to $S$,
  (2) deletion of an integer $x$ from $S$, and
  (3) searching for the integer $y \le x$ in $P$ for given $x$ such that $p(y)$ is maximized.
  For any instance $S$ of the priority searching problem, we can consider the path graph $G_S$ of size $|S|$ obtained
  by connecting the elements in $S$ linearly.
  The weight of each element is the priority of the element.
  Clearly, any query/update of the priority searching on $S$ can be simulated by
  a query/update of the dynamic SmQ on $G_S$.
\end{proof}

\section{Online/sliding rightmost LZ factorizations}

In this section, we present our algorithms for Problems (1)-(3).
We begin with our key data structure.

\subsection{BP-linked suffix trees}

We call the suffix tree of string $S$ augmented with its BP-representation
a \emph{BP-linked suffix tree}
and denote it by $\BPSTree(S)$. 
See Fig.~\ref{fig:BPSTree} for a concrete example of $\BPSTree(S)$.
Note that the BP-linked suffix tree is similar to
the \emph{timestamped suffix tree} proposed by Amir et al.~\cite{timestamped_suffixtree}.
However, the BP-linked suffix tree is superior to the timestamped suffix tree
in the following sense:
Our BP-linked suffix trees support a node deletion in worst-case $O(1)$ time,
while the timestamped suffix trees can require $\Omega(n)$ time for a node deletion in the worst case to maintain their \emph{rightmost/leftmost leaves pointers for all nodes}.

\begin{figure}[t]
  \centering  
  \includegraphics[width=0.9\linewidth]{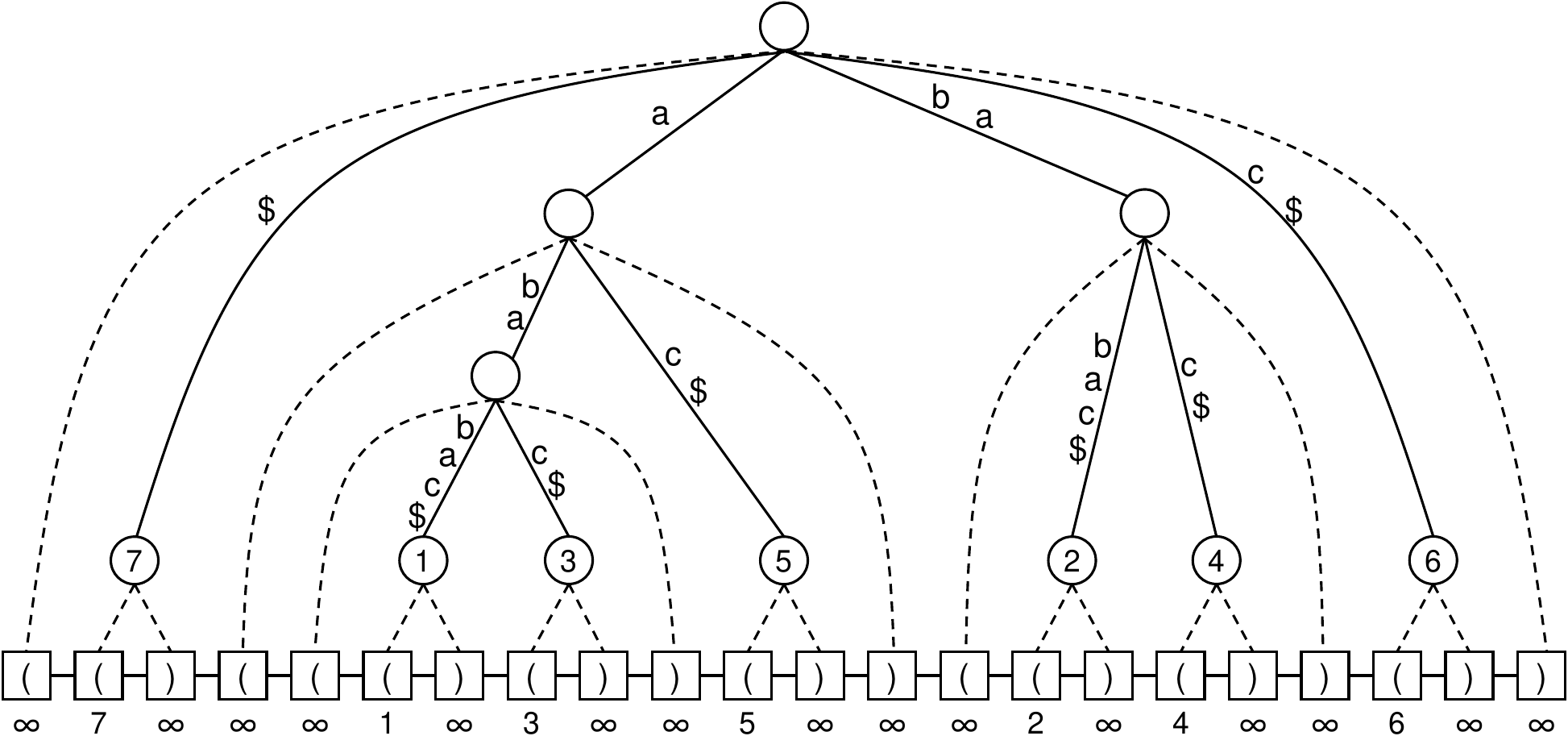}
  \caption{The BP-linked suffix tree of string $S = \mathtt{ababac\$}$.
  }\label{fig:BPSTree}
\end{figure}

By combining Lemma~\ref{lem:bptree} with the known
online suffix tree construction algorithms,
we immediately obtain the following results:

\begin{corollary} \label{coro:Weiner-type}
  For a string $S$ of length $n$,
  using $O(n)$ working space,
  one can update $\BPSTree(S)$ to $\BPSTree(cS)$
  and find the locus of $\lrp(cS)$ in $\BPSTree(cS)$
  for a given character $c \in \Sigma$
  \begin{enumerate}
    \item[(a)] in worst-case $O(\log\log n + (\log \log \sigma)^2 / \log \log \log \sigma)$ time for an integer alphabet of size $\sigma = n^{O(1)}$ with Fischer and Gawrychowski's algorithm~\cite{FischerG15,FischerG15_arxiv};
    \item[(b)] in amortized $O(\log \sigma)$ time for a general ordered alphabet of size $\sigma$ with Weiner's algorithm~\cite{Weiner73}.
  \end{enumerate}  
\end{corollary}

\begin{corollary} \label{coro:Ukkonen}
  For a string $S$ of length $n$ over a general ordered alphabet of size $\sigma$,
  using $O(n)$ working space,
  one can update $\BPSTree(S)$ to $\BPSTree(Sc)$ and find
  the locus of $\lrs(Sc)$ in $\BPSTree(Sc)$ for a given character $c \in \Sigma$
  in amortized $O(\log \sigma)$ time
  with Ukkonen's algorithm~\cite{Ukkonen95}.
\end{corollary}

Also, we employ our dynamic SmQ data structure~(Lemma~\ref{lem:smq}) to the BP-linked suffix trees.
This gives us the following:

\begin{lemma} \label{lem:BP_rmq}
  For an online string of length $n$,
  there exists a data structure of size $O(n)$ which supports,
  \begin{enumerate}
    \item[(a)] in worst-case $O(\log n/ \log \log n)$ time for an integer alphabet of size $\sigma = n^{O(1)}$ after $o(n)$-time preprocessing;
    \item[(b)] in amortized $O(\log \sigma + \log n/ \log \log n)$ time for a general ordered alphabet of size $\sigma$,
  \end{enumerate}
  the following queries and updates:
  \begin{itemize}
    \item Given an implicit or explicit node $v$ on the current suffix tree,
      find the leftmost occurrence of $\str(v)$ in the current string;
    \item Update the data structure when a new character is prepended.
  \end{itemize}
\end{lemma}
\begin{proof}
  Let $S$ be the input string.
  Since we use a Weiner-type of construction where a new character $c$ is prepended to $S$,
  we can assume that the right-end of $S$ terminates with a end-maker $\$$,
  with which all the suffixes of $S$ are represented by the leaves of $\STree(S)$.

  We consider Case (a).
  Let $\weight(v)$ be the weight of $v$ for each node $v$ in $\STree(S)$.
  For each leaf $\ell$, we set $\weight(\ell)$ to the beginning position of the suffix corresponding to $\ell$.
  For each non-leaf node $v$, we set $\weight(v) = \infty$.
  By applying Lemma~\ref{lem:smq} to this weighted suffix tree,
  we can answer the query in $O(\log n / \log \log n)$ time.
  Also, the auxiliary data structures can be updated in worst-case $O(\log n / \log \log n)$ time
  by Corollary~\ref{coro:Weiner-type}-(a) and Lemma~\ref{lem:smq}.

  Case (b) can be proven similarly with Corollary~\ref{coro:Weiner-type}-(b).
\end{proof}

\subsection{Online rightmost LPF}

Here we present our algorithm for Problem (1).

\begin{theorem}[Online rightmost LPF] \label{theo:LPF_online}
  For a string $S$ of length $n$,
  there exist online algorithms which use $O(n)$ space and compute $\RLPF[i]$ for each $1 \leq i \leq n$
  \begin{itemize}
    \item[(a)] in worst-case $O(\log n / \log \log n)$ time after $o(n)$-time preprocessing for an integer alphabet of size $\sigma = n^{O(1)}$;
    \item[(b)] in amortized $O(\log \sigma + \log n / \log \log n)$ time for a general order alphabet of size $\sigma$.
  \end{itemize}
\end{theorem}

\begin{proof}
  Let us consider Case (a).
  Since $\lrs(S[1..i]) = \lrp(\rev{(S[1..i])}) = \lrp(\rev{S}[n-i+1..n])$,
  the problem is reducible to computing
  the locus $p_j$ of $\lrp(\rev{S}[j..n])$ on $\STree(\rev{S}[j..n])$ for decreasing $j = n, \ldots, 1$,
  and finding the leaf in the subtree under $p_j$ that has the second smallest value.
  For this sake we can use (1) of Corollary~\ref{coro:Weiner-type} and Lemma~\ref{lem:BP_rmq}.
  Since $\sigma = n^{O(1)}$, we have $\log \log n + (\log \log \sigma)^2 / \log \log \log \sigma \in O(\log n / \log \log n)$.
  Thus $\RLPF[i]$ can be computed in worst-case $O(\log n / \log \log n)$ time each,
  after $o(n)$-time preprocessing. Case (b) can be shown similarly.
\end{proof}

\subsection{Online rightmost LZ-factorization}

In this subsection, we present our algorithm for Problem (2).

\begin{theorem}[Online rightmost LZ] \label{theo:online_LZ}
  For a string $S$ of length $n$ over a general order alphabet of size $\sigma$,
  there exists an online algorithm which uses $O(n)$ space and computes the rightmost LZ-factorization of $S$ 
  in amortized $O(\log \sigma + \log n / \log \log n)$ time per character.
\end{theorem}

\begin{proof}
  We use a standard technique with Ukkonen's online suffix tree construction
  with Corollary~\ref{coro:Ukkonen}.
  Suppose we have computed the first $k-1$ factors $f_{1}, \ldots, f_{k-1}$,
  and that we have built $\BPSTree(S[1..i])$ where $i = |f_1 \cdots f_{k-1}|+1$ is the beginning position of the next factor $f_k$.
  If $S[i]$ is a fresh character, then clearly $f_k = S[i]$.
  Otherwise, we perform the following.
  We grow the BP-liked suffix tree while reading subsequent characters $S[i+\ell-1]$ for increasing $\ell = 2, 3, \ldots$
  until we find the smallest
  $\ell^\star \ge 2$ such that $|\lrs(S[1.. i+\ell^\star-1])| < \ell^\star$~(see Fig.~\ref{fig:onlineLZ}).
  When we find such $\ell^\star$, it turns out that $f_k = S[i.. i+\ell^\star-2]$
  since $S[i.. i+\ell^\star-2]$ has a previous occurrence beginning at some position in $S[1.. i-1]$
  and $S[i.. i+\ell^\star-1]$ does not.
  Now, we search for the rightmost previous occurrence of $f_k$ by using $\BPSTree(S[1.. i+\ell^\star-1])$.
  Since $|\lrs(S[1.. i+\ell^\star-1])| \le \ell^\star-1 = |f_k|$,
  all the occurrences of $f_k$ are represented by leaves or the \emph{active point}
  that is the locus corresponding to the longest repeating suffix.
  Thus the rightmost previous occurrence of $f_k$ can be obtained by querying RMQs $O(1)$ times for the leaves under the locus of $f_k$.
  The above procedures for $f_k$ can be done in $O(|f_k| \log\sigma + \log n/\log \log n)$ time
  except for the time for maintaining the BP-linked suffix trees that takes $O(1)$ amortized time per character.
\end{proof}
\begin{figure}[t]
  \centering  
  \includegraphics[width=0.8\linewidth]{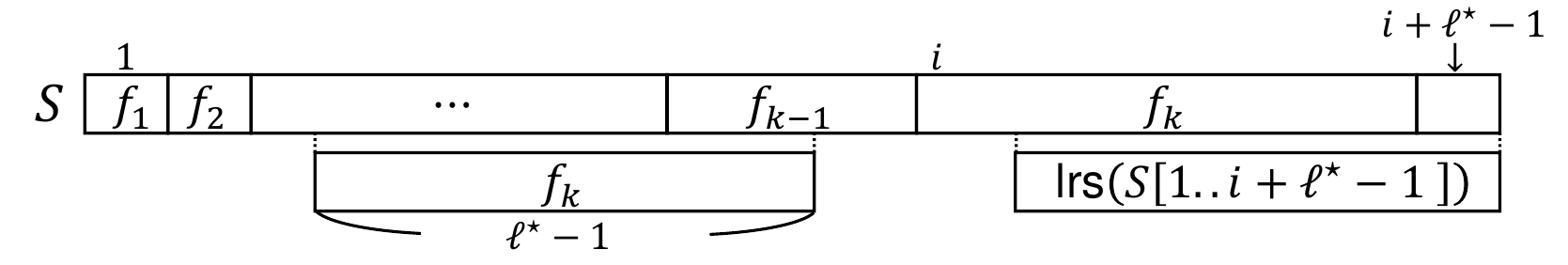}
  \caption{Illustration for Theorem~\ref{theo:online_LZ}.
  }\label{fig:onlineLZ}
\end{figure}

\subsection{Sliding rightmost LZ}
In this subsection, we present our algorithm for Problem (3).
\begin{theorem}[Sliding rightmost LZ] \label{thm:sliding_rightmostLZ}
  For an online string of length $n$ over a general ordered alphabet of size $\sigma$ and a fixed window size $d$, one can compute the sliding window rightmost LZ-factorization in amortized $O(\log \sigma + \log d / \log \log d)$ time per character, using $O(d)$ total space.
\end{theorem}

\begin{proof}
  We use a similar strategy to the case of online rightmost LZ-factorization from Theorem~\ref{theo:online_LZ},
  with a variant of Corollary~\ref{coro:Ukkonen} using a sliding suffix tree algorithm (cf.~\cite{Larsson96,Senft2005,leonard2024constanttime}).
  Suppose that we have computed the first $k-1$ factors $g_{1}, \ldots, g_{k-1}$, and
  that we have maintained $\BPSTree(W_{i})$ where $W_{i} = S[i-d.. i-1]$ is the current window of width $d$.
  If $S[i]$ does not occur in $W_i$,
  then clearly $g_k = S[i]$.
  Otherwise, as in Theorem~\ref{theo:online_LZ},
  we grow the BP-liked suffix tree while reading subsequent characters $S[i+\ell-1]$ for increasing $\ell = 2, 3, \ldots, 2d$
  until the value $\ell$ reaches $2d$ or
  we find the smallest
  $\ell^\star \ge 2$
  such that $|\lrs(S[i-d.. i+\ell^\star-1])| < \ell^\star$.
  If such $\ell^\star$ is found, then $g_k = S[i.. i+\ell^\star-2]$ and
  we can retrieve the rightmost previous occurrence of $g_k$ as in Theorem~\ref{theo:online_LZ}.
  Otherwise, $\ell = 2d$ and $|\lrs(S[i-d.. i+2d-1])| \ge 2d$ hold, and we then stop growing the suffix tree.
  Let $g' = S[i.. i+2d-1]$ be the length-$2d$ suffix of the extended window $S[i-d.. i+2d-1]$.
  Let $p$ be the difference between the beginning positions of the occurrence of $\lrs(S[i-d.. i+2d-1])$ as suffix and its (arbitrary) previous occurrence.
  Now $p \le d$ holds since $\lrs(S[i-d.. i+2d-1]) \ge 2d$.
  Then, $g'$ also appears $p$ positions to the left, i.e., at position $i-p$, and thus, $p$ is a period of $g'$ and $p \le |g'|/2$.
  The longest right-extension of $g'$ with period $p$ is $g_k$~(see Fig.~\ref{fig:slidingLZ}).
  Such extension can be computed in $O(|g_k|)$ time
  {with $O(d)$ space
  }by naive character comparisons in $S$
  {as follows:
    for incremental $j = 0, 1, 2, \ldots$,
    we compare character $S[i + 2d + j]$ to $S[i + (j \mod p)]$
    instead of $S[i + 2d + j - p]$
    until a mismatch is found.
    By doing this, no matter how large $j$ becomes,
    every character comparison is possible by retaining only
    the extended window $S[i-d.. i+2d-1]$ of size $3d$ and a single character $S[i + 2d + j]$.
  }

  At each $k$th step,
  we use only $O(d)$ space
  for the BP-linked suffix tree of an extended window of length at most $3d$
  and some auxiliary $O(1)$ working space.
  While we may need to compare $\omega(d)$ characters in $S$ when $g_k$ is much longer than $2d$,
  we do not need to store the characters outside of the extended window. Thus, such character-comparisons can be done within $O(d)$ space.
  {Then, to proceed to the $(k+1)$th step,
    we move to the next window of size $d$,
    namely, the length-$d$ suffix of $S[1.. |g_1g_2\cdots g_k|]$.
  }\begin{figure}[t]
    \centering  
    \includegraphics[width=0.8\linewidth]{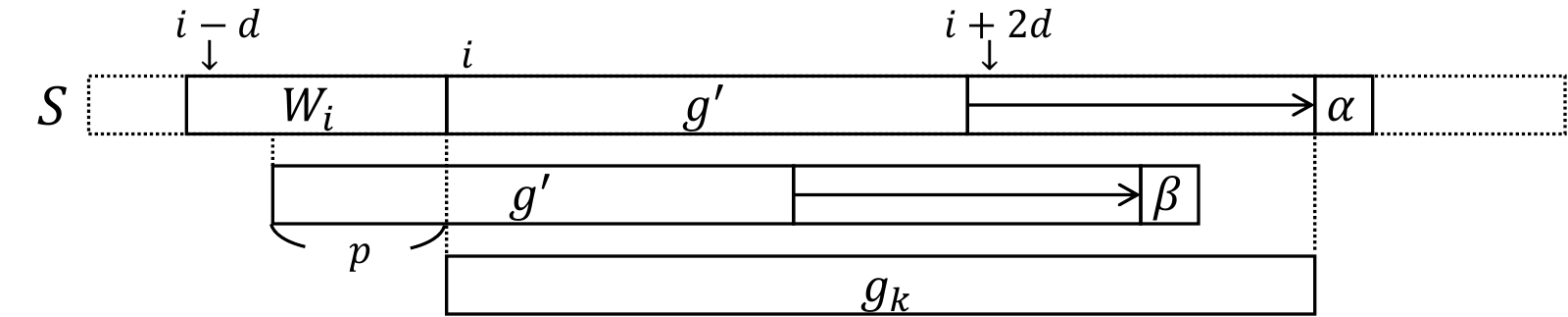}
    \caption{Illustration for Theorem~\ref{thm:sliding_rightmostLZ}.
      String $g'$ of length $2d$ has period $p$.
      When $\alpha \ne \beta$ where $\beta$ is $p$ characters before $\alpha$,
      the next factor $g_k$ is determined
      since $g_k\alpha$ cannot occur before it due to the periodicity of $g_k\beta$.
    }\label{fig:slidingLZ}
  \end{figure}
\end{proof}
\section{Other applications of BP-linked suffix trees} \label{sec:applications}

In this section, we present other applications of our BP-linked (suffix) trees,
which are online computation of \emph{closed factorizations} of a given string.

\subsection{Online longest closed factorizations}
A string $w$ is \emph{closed} if $w$ is a character,
or the longest border $b$ of $w$ occurs exactly twice in $w$ as prefix and suffix~\cite{Fici11}.
The \emph{longest closed factorization}
$\LCF(S) = g_1, \ldots, g_k$ of a string $S$
is a factorization of $S$ such that 
each $g_i$ is the longest closed suffix of $S[1..|g_1 \cdots g_i|]$.
The \emph{longest closed factor array} $\LCFA$ of a string $S$ of length $n$
is an array of length $n$ such that 
$\LCFA[i]$ stores the length of the last factor of 
$\LCF({S[1..i]})$ and the size of $\LCF(S[1..i])$ for $1 \leq i \leq n$.
$\LCF(S)$ can readily be obtained from $\LCFA$ for $S$.

Alzamel et al.~\cite{AlzamelISS19} showed the following property:
\begin{lemma}[\cite{AlzamelISS19}] \label{lem:lrs_LCF}
  For a string $S$, if $g_1, \ldots, g_k = \LCF(S)$,
  then $g_k = S[i..|S|]$,
  where $i$ is the second rightmost occurrence of $\lrs(S)$ in $S$.
  Also, $\lrs(S)$ is the longest border of $g_k$.
\end{lemma}

Alzamel et al.~\cite{AlzamelISS19} employ Ukkonen's online suffix tree and rely on RMQ on a dynamic list of leaves, for computing $\LCFA$ online.
The inputs of their RMQ is given as a pair $l,r$ of two integers representing
an interval $[l, r]$ in the sorted list of leaves in the online suffix tree,
where $l$ and $r$ are the lexicographical ranks of the leftmost and rightmost leaves
in the subtree rooted at the active point.
However, in~\cite{AlzamelISS19} the authors do not describe how to explicitly
maintain the ranks of leaves on a growing suffix tree as integers.
We remark that even a single leaf insertion to the suffix tree
can change the ranks of $\Omega(n)$ existing leaves.

However, as we have observed previously,
by the use of our online BP-linked suffix tree,
maintaining the ranks of the leaves in a growing suffix tree is no more necessary for performing RMQs under the active point.
Due to Lemma~\ref{lem:lrs_LCF},
we can use a similar strategy as in Theorem~\ref{theo:LPF_online}
by noting that the second rightmost occurrence,
which is the second leftmost occurrence in the reversed string,
can be found with a constant number of RmQs.
Thus we have:

\begin{theorem} \label{theo:LCFA_online}
  For a string $S$ of length $n$,
  there exist online algorithms which use $O(n)$ space and compute $\LCFA[i]$ for each $1 \leq i \leq n$
  \begin{itemize}
    \item[(a)] in worst-case $O(\log n / \log \log n)$ time after $o(n)$-time preprocessing for an integer alphabet of size $\sigma = n^{O(1)}$;
    \item[(b)] in amortized $O(\log \sigma + \log n / \log \log n)$ time for a general order alphabet of size $\sigma$.
  \end{itemize}
\end{theorem}
Our result in Theorem~\ref{theo:LCFA_online} can be seen as an online alternative to the offline solution in the literature~\cite{BadkobehBGIIIPS16}, with the same complexity.

\subsection{Online minimum closed factorizations}
The closed factorization $g_1, \ldots, g_k$ of a string $S$ 
is called the \emph{minimum closed factorization} of $S$
if the number $k$ of factors is smallest~\cite{BadkobehBGIIIPS16}.
Let $\mcf(S)$ denote the size of the minimum closed factorization of $S$.

\begin{theorem} \label{theo:MCFA_online}
  For a string $S$ of length $n$,
  there exist online algorithms which use $O(n)$ space and compute
  the minimum closed factor array $\MCFA[i] = \mcf(S[1..i])$ for each $1 \leq i \leq n$, with $\ell_i = |\lrs(S[1..i])|$,
  \begin{itemize}
    \item[(a)] in worst-case $O(\ell_i \log n / \log \log n)$ time after $o(n)$-time preprocessing for an integer alphabet of size $\sigma = n^{O(1)}$;
    \item[(b)] in amortized $O(\log \sigma + \ell_i \log n / \log \log n)$ time for a general order alphabet of size $\sigma$.
  \end{itemize}
\end{theorem}

\begin{proof}
  Consider Case (a).
  We find the locus for $\lrp(\rev{S[1..i]})$ in $\BPSTree(\rev{S[1..i]})$
  in worst-case $O(\log n / \log \log n)$ time with Corollary~\ref{coro:Weiner-type}-(a).
  Let $v_1, \ldots, v_{\ell_i}$ be the explicit/implicit nodes on the path from the root to the locus for $\lrp(\rev{S[1..i]})$.
  For each $v_j$, we perform a constant number of RmQs to find the second leftmost
  occurrence of $\str(v_j)$ with Lemma~\ref{lem:BP_rmq}
  in worst-case $O(\log n / \log \log n)$ time.
  Then, we can compute $\mcf(\rev{S[1..i]})[i]$ by dynamic programming in $O(\ell_i)$ time.

  Case (b) can be obtained with Corollary~\ref{coro:Weiner-type}-(b).
\end{proof}

Alzamel et al.~\cite{AlzamelISS19} claimed a solution with 
$O(\ell_i (\log \sigma + \log n))$ worst-case running time
for each $i$, which is based on Ukkonen's algorithm.
Although amortized, our algorithm is faster than theirs also in the case of general ordered alphabets.

\subsubsection{Acknowledgments}
This work was supported by JSPS KAKENHI Grant Numbers JP23H04381, JP24K20734~(TM)
and JP20H05964, JP23K24808, JP23K18466~(SI).
{The authors thank the anonymous referee for a pointer to reference~\cite{HenzingerK99} that introduced the Euler tour trees.
}

\end{document}